\documentclass[conference]{IEEEtran}

\renewcommand{\baselinestretch}{0.97}

\usepackage{graphicx,colordvi,psfrag}
\usepackage{amsmath,amssymb}
\usepackage{epstopdf}
\usepackage[caption=false]{subfig}
\usepackage{epsfig,cite}
\usepackage{calc,pstricks, pgf, xcolor}
\usepackage{nicefrac}
\usepackage{enumerate}
\usepackage{dsfont}
\usepackage{bm}
\allowdisplaybreaks
\newtheorem{theorem}{Theorem}
\newtheorem{corollary}{Corollary}
\newtheorem{remark}{Remark}

\newtheorem{definition}{Definition}
\newtheorem{proposition}{Proposition}
\newtheorem{lemma}{Lemma}
\newtheorem{conjecture}{Conjecture}
\newenvironment{proof}[1][Proof]{\noindent\textbf{#1.} }{\ \rule{0.5em}{0.5em}}
\newcommand{\Ind}{\mathds{1}}

\newcommand{\Ber}{\mathrm{Bernoulli}}

\newcommand{\m}{\mathcal}

\renewcommand{\thefigure}{\arabic{figure}}
\newcommand{\p}[1]{\left(#1\right)}
\newcommand{\pp}[1]{\left[#1\right]}
\newcommand{\ppp}[1]{\left\{#1\right\}}
\newcommand{\abs}[1]{\left|#1\right|}

\begin{document}

\title{How to Quantize $n$ Outputs of a Binary Symmetric Channel to $n-1$ Bits?}

\author{\IEEEauthorblockN{Wasim Huleihel}
\IEEEauthorblockA{MIT\\
wasimh@mit.edu}
\and
\IEEEauthorblockN{Or Ordentlich}
\IEEEauthorblockA{MIT\\
ordent@mit.edu}
\IEEEoverridecommandlockouts
\IEEEcompsocitemizethanks{
\IEEEcompsocthanksitem
The work of W. Huleihel and O. Ordentlich was supported by the MIT - Technion Postdoctoral Fellowship.
}}

\parskip 3pt

\maketitle

\begin{abstract}
Suppose that $Y^n$ is obtained by observing a uniform Bernoulli random vector $X^n$ through a binary symmetric channel with crossover probability $\alpha$. The ``most informative Boolean function'' conjecture postulates that the maximal mutual information between $Y^n$ and any Boolean function $\mathrm{b}(X^n)$ is attained by a dictator function. In this paper, we consider the ``complementary" case in which the Boolean function is replaced by $f:\ppp{0,1}^n\to\ppp{0,1}^{n-1}$, namely, an $n-1$ bit quantizer, and show that $I(f(X^n);Y^n)\leq (n-1)\cdot\left(1-h(\alpha)\right)$ for any such $f$. Thus, in this case, the optimal function is of the form $f(x^n)=(x_1,\ldots,x_{n-1})$.
\end{abstract}

\section{Introduction}\label{sec:intro}

Let $X^n$ be an $n$-dimensional binary vector uniformly distributed over $\ppp{0,1}^n$, and let $Y^n$ be the output of passing $X^n$ through a binary symmetric channel (BSC) with crossover probability $\alpha\in[0,1/2]$. In other words, $Y^n = X^n\oplus Z^n$, where $Z^n$ is a sequence of $n$ independent and identically distributed (i.i.d.) $\mathrm{Bernoulli}(\alpha)$ random variables, statistically independent of $X^n$. The following conjecture~\cite{ck14} have recently received considerable attention.
\begin{conjecture}\label{conj:1}
For any \emph{Boolean} function $\mathrm{b}:\ppp{0,1}^n\to\ppp{0,1}$, we have $I(\mathrm{b}(X^n);Y^n)\leq1-h(\alpha)$, where $h(\alpha)\triangleq -\alpha\log_2\alpha-(1-\alpha)\log_2(1-\alpha)$ is the binary entropy function.
\end{conjecture}

Since the dictator function $\mathrm{b}(X^n)=X_i$ (for any $1\leq i\leq n$) achieves this upper bound with equality, then intuitively Conjecture~\ref{conj:1} postulates that the dictator function is the most ``informative" one-bit quantization of $X^n$ in terms of achieving the maximal $I(\mathrm{b}(X^n);Y^n)$. Clearly, by the symmetry of the pair $(X^n,Y^n)$ we have that for any function $I(\mathrm{b}(X^n);Y^n)=I(X^n;\mathrm{b}(Y^n))$, so we can equivalently think of the problem at hand as seeking the optimal one-bit quantizer of $n$ outputs of the channel.
Despite attempts in various directions \cite{ck14,agkn13,ct14,osw16,kow15,ns16,pmp16}, Conjecture~\ref{conj:1} remains open in general. However, for the ``very noisy'' case, where $\alpha>1/2-\delta$, for some $\delta>0$ independent of $n$, the validity of the conjecture was established by Samorodnitsky~\cite{Samorodnitsky16}.

In this paper, we consider the ``complementary" case in which the Boolean function in Conjecture~\ref{conj:1} is replaced by an $n-1$ bit quantizer. Our main result is the following.

\begin{theorem}
For any function $f:\ppp{0,1}^n\to\ppp{0,1}^{n-1}$ we have
\begin{align}
I(f(X^n);Y^n)\leq (n-1)\cdot\left(1-h(\alpha)\right),\label{ToShow}
\end{align}
and this bound is attained with equality by, e.g., $f(x^n)=(x_1,\ldots,x_{n-1})$.
\label{thm:main}
\end{theorem}

One may wonder whether for any $f:\ppp{0,1}^n\to\ppp{0,1}^{k}$ we have $I(f(X^n);Y^n)\leq k\cdot\left(1-h(\alpha)\right)$. However, for $k=Rn$ with $0<R<1$, the problem essentially reduces to remote source coding under log-loss distortion measure, for which the maximal value of $I(f(X^n);Y^n)/n$ (as a function of $R,\alpha$) can be determined up to $o(n)$ terms. Indeed, \cite{ErkipCover,ct14} characterizes this quantity which turns out to be greater than $R\cdot(1-h(\alpha))$. Conjecture~\ref{conj:1} as well as Theorem~\ref{thm:main} deal with the extreme cases of $k=1$ and $k=n-1$, respectively, where neglecting the $o(n)$ terms leads to non-informative characterization of the maximal $I(f(X^n);Y^n)$, and therefore \cite{ct14,ErkipCover} do not suffice.

Theorem~\ref{thm:main} can be generalized to a stronger statement concerning the entire class of binary-input memoryless output-symmetric (BMS) channels.

\begin{definition}[BMS channels]
A memoryless channel with binary input $X$ and output $Y$ is called \emph{binary-input memoryless output-symmetric (BMS)} if there exists a sufficient statistic $g(Y)=(X\oplus Z_T,T)$ for $X$, where $(T,Z_T)$ are statistically independent of $X$, and $Z_T$ is a binary random variable with $\Pr(Z_T=1|T=t)=t$.
\end{definition}

\begin{corollary}[\cite{yp_pc}]
Let $X^n$ be an $n$-dimensional binary vector uniformly distributed over $\ppp{0,1}^n$, and let $Y^n$ be the output of passing $X^n$ through a BMS with capacity $C$. Then for every $f:\ppp{0,1}^n\to\ppp{0,1}^{n-1}$, we have $$I(f(X^n);Y^n)\leq (n-1)\cdot C,$$ and this bound is attained with equality by, e.g., $f(x^n)=(x_1,\ldots,x_{n-1})$.
\label{cor:bms}
\end{corollary}

\begin{proof}[Proof of Corollary~\ref{cor:bms}]
Let $W$ be a BMS channel with capacity $C=1-h(\alpha)$.
Let $Y^n_{W}$ and $Y^n_{\mathrm{BSC}}$ be the outputs corresponding to the channel $W$ and a BSC with crossover probability $\alpha$, respectively, when the input to both channels is $X^n$. Define $\pp{n}\triangleq\ppp{1,2,\ldots,n}$. For any function $f:\{0,1\}^n\to [M]$, we can write
\begin{align}
&I(f(X^n);Y_W^n)=I(f(X^n),X^n;Y_W^n)-I(X^n;Y_W^n|f(X^n))\nonumber\\
&=I(X^n;Y_W^n)+I(f(X^n);Y_W^n|X^n)-I(X^n;Y_W^n|f(X^n))\nonumber\\
&=I(X^n;Y_W^n)\nonumber\\
&\ \ \ -\sum_{m=1}^M \Pr(f(X^n)=m)I(X^n;Y_W^n|f(X^n)=m).\nonumber
\end{align}
We proceed by noting that $I(X^n;Y_W^n)=I(X^n;Y_{\mathrm{BSC}}^n)=nC$ as the capacity achieving input distribution of both channels is $\Ber(1/2)$. Furthermore, recall the fact that the BSC is the least capable among all BMS channels with the same capacity~\cite[page 116]{csiszarkorner},\cite[Lemma 7.1]{SasogluPhD}. To wit, for any input $U^n$, the corresponding outputs of $W$ and the BSC will satisfy $I(U^n;Y^n_{\mathrm{BSC}})\leq I(U^n;Y_W^n)$. This implies that 
\begin{align}
I(X^n;Y_W^n|f(X^n)=m)\geq I(X^n;Y_{\mathrm{BSC}}^n|f(X^n)=m),\nonumber
\end{align}
for all $m=1,\ldots,M$. Thus, we get that for any function $f$, 
\begin{align}
I(f(X^n);Y_W^n)\leq I(f(X^n);Y_{\mathrm{BSC}}^n).
\end{align}
The corollary now follows by invoking Theorem~\ref{thm:main}.
\end{proof}

\section{Proof of Theorem~\ref{thm:main}}\label{sec:proof}

Since the vector $Y^n$ is uniformly distributed over $\ppp{0,1}^n$, we have
\begin{align}
I(f(X^n);Y^n)&=n-H(Y^n|f(X^n)).\label{eq:MI}
\end{align}
Our goal is therefore to lower bound $H(Y^n|f(X^n))$.

Consider the function $f:\{0,1\}^n\to[2^{n-1}]$, and define the sets
\begin{align*}
f^{-1}(j)\triangleq\ppp{x^n\in\ppp{0,1}^n:\;f(x^n)=j}, \ j=1,\ldots,2^{n-1},
\end{align*}
which form a disjoint partition of $\{0,1\}^n$. Further, define the sizes of these sets as
\begin{align*}
m_j\triangleq \abs{f^{-1}(j)} = \sum_{x^n\in\ppp{0,1}^n}\Ind\ppp{f(x^n)=j}, \ j=1,\ldots,2^{n-1},
\end{align*}
and assume without loss of generality that $m_j>0$, for all $j$. To see why this assumption is valid, first note that there must exist some $i$, for which $m_i\geq 2$. Let $f^{-1}(i)=\{x_{i_1}^n,\ldots,x_{i_{m_i}}^n\}$. Now if there exists some $j\neq i$, such that $m_j=0$, we can define a new function $\tilde{f}:\{0,1\}^n\to[2^{n-1}]$ where $\tilde{f}^{-1}(i)=\{x_{i_1}^n,\ldots,x_{i_{m_i-1}}^n\}$, $\tilde{f}^{-1}(j)=\{x_{i_{m_i}}^n\}$, and $\tilde{f}^{-1}(t)=f^{-1}(t)$, for all $t\neq i,j$. For this function we must have
\begin{align}
H\left(Y^n|f(X^n)\right)&\geq H\left(Y^n|f(X^n),\Ind\ppp{X^n=x_{i_{m_i}}}\right)\nonumber\\
&=H(Y^n|\tilde{f}(X^n)),\nonumber
\end{align}
and consequently $I(f(X^n);Y^n)\leq I(\tilde{f}(X^n);Y^n)$.

Next, for every $m=0,1,\ldots, 2^n$ define the quantity
\begin{align}
\lambda(m)\triangleq\sum_{j=1}^{2^{n-1}}\Ind\ppp{m_j=m},\label{lamdef}
\end{align}
which counts the number of sets $f^{-1}(j)$ with cardinality $m$, in the partition induced by the function $f$.\footnote{In fact, since we have already assumed that $m_j>0$ for all $j$, we have that $\lambda(0)=0$ and $\lambda(m)=0$ for $m>2^n-(2^{n-1}-1)$.} The next proposition expresses $\lambda(1)$ in terms of $\{\lambda(m)\}_{m\geq 2}$.
\begin{proposition}
For any $f:\ppp{0,1}^n\to[2^{n-1}]$ with $m_j>0$ for all $j$, we have that
\begin{align}
\lambda(1)=\sum_{m\geq 3} (m-2)\lambda(m).
\end{align}
\label{prop:lambda1}
\end{proposition}
\vspace{-0.3cm}

Intuitively, this proposition states that since the average size of the sets $f^{-1}(j)$ is $2$, then every set $f^{-1}(j)$ of cardinality $m>2$, must be compensated for by $(m-2)$ sets of cardinality $1$.

\begin{proof}
Using the definition of $\lambda(m)$ in \eqref{lamdef}, and the fact that $\{f^{-1}(j)\}$ forms a disjoint partition of $\ppp{0,1}^n$, we have
\begin{align}
\sum_{m=0}^{2^n}\lambda(m)&=2^{n-1},\label{eq:lambdasum}\\
\sum_{m=0}^{2^n}m\lambda(m)&=2^{n}.\label{eq:mlambdasum}
\end{align}
Multiplying~\eqref{eq:lambdasum} by $2$ and equating it with the left-hand side of~\eqref{eq:mlambdasum}, we get
\begin{align}
\sum_{m=0}^{2^n} 2\lambda(m)=\sum_{m=0}^{2^n}m\lambda(m),\nonumber
\end{align}
which implies
\begin{align}
2\lambda(0)+\lambda(1)=\sum_{m\geq 3}(m-2)\lambda(m).\nonumber
\end{align}
Invoking our assumption that $\lambda(0)=0$ gives the desired result.
\end{proof}

\begin{definition}[Minimal entropy of a noisy subset]
For a family of vectors $S\subset \ppp{0,1}^n$ let $U_S$ be a random vector uniformly distributed over $S$, and let $Z^n$ be a sequence of $n$ i.i.d. $\mathrm{Bernoulli}(\alpha)$ random variables, statistically independent of $U_S$. For $m=1,\ldots,2^n$, we define the quantity
\begin{align}
H_m^n(\alpha)\triangleq \min_{S\subset\ppp{0,1}^n \ : \ |S|=m}H(U_S\oplus Z^n).
\label{eq:Hmdef}
\end{align}
\end{definition}

Some properties of $H_m^n(\alpha)$ will be studied in the next section. In particular, we will prove the following lemma.

\begin{lemma}
For any $2<m<2^n$,
\begin{align}
\frac{m-2}{2m-2}H_1^n(\alpha)+\frac{m}{2m-2}H_m^n(\alpha)\geq H_2^n(\alpha).\label{H2Boundprove}
\end{align}
\label{lem:H2bound}
\end{lemma}

We can now write
\begin{align}
&H(Y^n|f(X^n))=\sum_{j=1}^{2^{n-1}}\Pr\left(f(X^n)=j\right)H\left(Y^n|f(X^n)=j\right)\nonumber\\
&=\sum_{j=1}^{2^{n-1}}\Pr\left(X^n\in f^{-1}(j)\right)H\left(Y^n|X^n\in f^{-1}(j)\right)\nonumber\\
&=2^{-n}\sum_{j=1}^{2^{n-1}}\abs{f^{-1}(j)}H\left(U_{f^{-1}(j)}\oplus Z^n\right)\nonumber\\
&\geq 2^{-n}\sum_{j=1}^{2^{n-1}}m_j H_{m_j}^n(\alpha)\nonumber\\
&= 2^{-n}\sum_{m=1}^{2^n}m\lambda(m)H_m^n(\alpha)\nonumber
\end{align}
\begin{align}
&=2^{-n}\left(\lambda(1)H_1^n(\alpha)+ 2\lambda(2)H_2^n(\alpha)+\sum_{m\geq 3}^{2^n}m\lambda(m)H_m^n(\alpha)\right)\nonumber\\
&=2^{-n}\bigg(2\lambda(2)H_2^n(\alpha)\nonumber\\
& \ \ \ \ \ \ +\sum_{m\geq 3}^{2^n}(m-2)\lambda(m)H_1^n(\alpha)+m\lambda(m)H_m^n(\alpha)\bigg)\label{eq:lambda1}\\
&=2^{-n}\bigg(2\lambda(2)H_2^n(\alpha)\nonumber\\
& +\sum_{m\geq 3}^{2^n}(2m-2)\lambda(m)\left[\frac{m-2}{2m-2}H_1^n(\alpha)+\frac{m}{2m-2}H_m^n(\alpha)\right]\bigg)\nonumber\\
&\geq H_2^n(\alpha)\cdot 2^{-n}\sum_{m=1}^{2^n} (2m-2)\lambda(m)\label{eq:H2bound}\\
&=H_2^n(\alpha),\label{eq:lambdasums}
\end{align}
where in~\eqref{eq:lambda1} follows from Proposition~\ref{prop:lambda1}, in~\eqref{eq:H2bound} we have used Lemma~\ref{lem:H2bound}, and~\eqref{eq:lambdasums} follows from~\eqref{eq:lambdasum} and~\eqref{eq:mlambdasum}. Proposition~\ref{prop:H2}, stated and proved in the next section, shows that $H_2^n(\alpha)=1+(n-1)h(\alpha)$. Combining this with~\eqref{eq:MI} and~\eqref{eq:lambdasums} establishes the desired result.

\section{Properties of $H_m^n(\alpha)$}\label{sec:properties}

The main goal of this section is to prove Lemma~\ref{lem:H2bound}. To this end, we establish some properties of the function $H_m^n(\alpha)$, which may be of independent interest.

\begin{proposition}[Monotonicity in $m$]
The function $H_m^n(\alpha)$ is monotonically non-decreasing as a function of $m$.
\end{proposition}

\begin{proof}
It is suffice to show that for any natural number $1\leq m<2^n$ it holds that $H_{m}^n(\alpha)\leq H_{m+1}^n(\alpha)$. To this end, let $S=\{s_1,\ldots,s_{m+1}\}\subset\ppp{0,1}^n$ be a family of ${m+1}$ vectors, and let $S_{-i}\triangleq S\setminus\{s_i\}$, for $i=1,\ldots,m+1$. Clearly, $|S_{-i}|=m$ for all $i$. Furthermore, the random vector $U_S$ can be generated by first drawing a random variable $A\sim\mathrm{Uniform}([m+1])$ and then drawing a statistically independent random vector uniformly over $S_{-A}$. Thus, for any $S\subset\ppp{0,1}^n$ of size $m+1$ we have that
\begin{align}
H(U_{S}\oplus Z^n)&\geq H\left(U_{S}\oplus Z^n|A\right)\nonumber\\
&=\frac{1}{m+1}\sum_{a=1}^{m+1}H\left(U_{S_{-a}}\oplus Z^n\right)\geq H_m^n(\alpha),\nonumber
\end{align}
and in particular $H_{m}^n(\alpha)\leq H_{m+1}^n(\alpha)$.
\end{proof}

We define the partial order ``$\leq$" on the hypercube $\{0,1\}^n$ as $y\leq x$ iff $y_i\leq x_i$, for all $i=1,\ldots,n$.

\begin{definition}[Monotone sets]
A set $S\subset\{0,1\}^n$ is monotone if $x\in S$ implies $y\in S$, for all $y\leq x$.
\end{definition}

Let $\m{M}_m^n\triangleq\left\{S\subset\ppp{0,1}^n: \ |S|=m, \ S \text{ is monotone}\right\}$. We will prove the following result.
\begin{lemma}[Sufficiency of monotone sets]
\begin{align}
H_m^n(\alpha)=\min_{S\in \m{M}_m^n} H(U_S\oplus Z^n).\nonumber
\end{align}
\label{lem:monsuffice}
\end{lemma}
\vspace{-0.3cm}
\begin{remark}
Theorem 3 in~\cite{ck14} states that among all boolean functions, $I(\mathrm{b}(X^n);Y^n)$ is maximized by functions for which the induced set $\mathrm{b}^{-1}(0)$ is monotone.\footnote{In fact,~\cite[Theorem 3]{ck14} provides a stronger statement about the structure of the induced $\mathrm{b}^{-1}(0)$.} While this statement is closely related to our Lemma~\ref{lem:monsuffice}, it does not imply it, although the proof technique is somewhat similar.
\end{remark}

The proof of Lemma~\ref{lem:monsuffice} is based on applying a procedure called \emph{shifting}~\cite{Kleitman66,Alon83,Frankl83}.

\begin{definition}[Shifting]\label{def:shifting}
For a set of binary vectors $S\subset\ppp{0,1}^n$ the \emph{shifting} procedure is defined as follows. For $i\in[n]$ and $x\in\{0,1\}^n$ write $x-i$ for the vector obtained by setting $x_i=0$, and define
\begin{align}
S_i\triangleq \{x\in S: \ x_i=1, x-i\notin S \}.\nonumber
\end{align}
Find the smallest $i$ such that $S_i\neq\emptyset$. If there is no such $i$ then we are done. Otherwise, replace $S$ with the set $(S\setminus S_i)\cup (S_i-i)$, where $S_i-i\triangleq \{x-i \ : x\in S_i\}$, and repeat. The output of this process is a monotone set, denoted by $S_{\text{shifted}}$, with cardinality $|S_{\text{shifted}}|=|S|$.
\end{definition}

The proof of Lemma~\ref{lem:monsuffice} hinges on the following result.

\begin{lemma}
Let $S\subset\{0,1\}^n$ be some subset of vectors, and $\bar{S}\subset\{0,1\}^n$ be the result of applying one iteration of the shifting procedure, say, on the first coordinate. Let $P_{Y|X}$ be some discrete memoryless channel with binary input, and let $Y^n$ be its output when the input is $U_S$ and ${\bar{Y}}^{n}$ be its output when the input is $U_{\bar{S}}$. For every $\omega\in\m{Y}^{n-1}$ we have that $\Pr(Y_2^n=\omega)=\Pr({\bar{Y}}_2^n=\omega)$, and
\begin{align}
\left|\Pr(U_{\bar{S},1}=1|{\bar{Y}}_2^n=\omega)-\frac{1}{2}\right|\geq \left|\Pr(U_{S,1}=1|Y_2^n=\omega)-\frac{1}{2}\right| .\nonumber
\end{align}
\label{lem:morebiass}
\end{lemma}

\begin{proof}[Proof of Lemma~\ref{lem:morebiass}]
Let $S_2^n$ be the projection of $S$ onto the coordinates $\{2,\ldots,n\}$, and note that the projection of $\bar{S}$ onto these coordinates is also $S_2^n$, as the shifting operations does not effect these coordinates. Consequently, $U_{S,2}^n$ and ${\bar{U}}_{S,2}^n$ have the same distribution, and therefore $Y_2^n$ and ${\bar{Y}}_2^n$ have the same distribution.

Next, for any vector $\omega\in\m{Y}^{n-1}$, we have
\begin{align}
\Pr&(U_{S,1}=1|Y_2^n=\omega)\nonumber\\
&=\sum_{x\in S_2^n} \Pr(U_{S,1}=1,U_{S,2}^n=x|Y_2^n=\omega)\nonumber\\
&=\sum_{x\in S_2^n} \Pr(U_{S,1}=1|U_{S,2}^n=x)\Pr(U_{S,2}^n=x|Y_2^n=\omega).\nonumber
\end{align}
The fact that $U_{S,2}^n$ and ${U}_{\bar{S},2}^{n}$ have the same distribution, implies that $P_{U_{S,2}^n|Y_2^n}=P_{U_{\bar{S},2}^n|{\bar{Y}}_2^n}$, and therefore
\begin{align}
\Pr&(U_{\bar{S},1}=1|{\bar{Y}}_2^n=\omega)\nonumber\\
&=\sum_{x\in S_2^n} \Pr(U_{\bar{S},1}=1|{U}_{\bar{S},2}^n=x)\Pr(U_{S,2}^n=x|Y_2^n=\omega).\nonumber
\end{align}
We partition the set $S_2^n$ into three subsets:
\begin{itemize}
\item $A\triangleq\{x\in S_2^n: \ [0 \ x]\in S, [1 \ x]\in S\}$
\item $B\triangleq\{x\in S_2^n: \ [0 \ x]\notin S, [1 \ x]\in S\}$
\item $C\triangleq\{x\in S_2^n: \ [0 \ x]\in S, [1 \ x]\notin S\}$
\end{itemize}
and we note that
\begin{align}
\Pr(U_{S,1}=1|U_{S,2}^n=x)=
\begin{cases}
1/2 & x\in A\\
1 & x\in B\\
0 & x\in C
\end{cases}.\nonumber
\end{align}
Letting
\begin{align}
a_\omega&\triangleq\Pr(U_{S,2}^n\in A|Y_2^n=\omega),\nonumber\\
b_\omega&\triangleq\Pr(U_{S,2}^n\in B|Y_2^n=\omega), \nonumber\\
c_\omega&\triangleq\Pr(U_{S,2}^n\in C|Y_2^n=\omega),\nonumber
\end{align}
we get
\begin{align}
\Pr(U_{S,1}=1|Y_2^n=\omega)=\frac{a_\omega}{2}+b_\omega\nonumber.
\end{align}
By the definition of the shifting procedure in Definition~\ref{def:shifting},
\begin{align}
\Pr(U_{\bar{S},1}=1|U_{\bar{S},2}^n=x)=
\begin{cases}
1/2 & x\in A\\
0 & x\in B\\
0 & x\in C
\end{cases}.\nonumber
\end{align}
Thus,
\begin{align}
\Pr(U_{\bar{S},1}=1|\bar{Y}_2^n=\omega)=\frac{a_\omega}{2}\nonumber.
\end{align}
We can use this to see that $\Pr(U_{\bar{S},1}=1|{\bar{Y}}_2^n=\omega)$ is more biased than $\Pr(U_{S,1}=1|Y_2^n=\omega)$. Indeed
\begin{align}
&\bigg(\frac{1}{2}-\Pr(U_{\bar{S},1}=1|{\bar{Y}}_2^n=\omega) \bigg)^2\nonumber\\
&\ \ \ \ \ \ -\left(\frac{1}{2}-\Pr(U_{S,1}=1|Y_2^n=\omega) \right)^2\nonumber\\
&=\left(\frac{1}{2}(1-a_\omega) \right)^2-\left(\frac{1}{2}(1-a_\omega)-b_\omega \right)^2\nonumber\\
&=b_\omega(1-a_\omega)-b_\omega^2=b_{\omega}c_\omega\geq 0,\nonumber
\end{align}
as desired.
\end{proof}

\begin{corollary}[Shifting decreases output entropy]
Let $S\subset\{0,1\}^n$ be some subset of vectors, and $\bar{S}\subset\{0,1\}^n$ be the result of applying one iteration of the shifting procedure, say, on the first coordinate. Let $Z^n$ be a sequence of $n$ i.i.d. $\mathrm{Bernoulli}(\alpha)$ random variables, statistically independent of $U_S$ and $U_{\bar{S}}$. Then,
\begin{align}
H(U_{\bar{S}}\oplus Z^n)\leq H(U_{S}\oplus Z^n).
\end{align}
\label{ref:shiftiteration}
\end{corollary}
\vspace{-0.4cm}
\begin{proof}
By the chain rule,
\begin{align}
&H(U_{S}\oplus Z^n)=H(U_{S,2}^n\oplus Z_{2}^n)+H(U_{S,1}\oplus Z_1|U_{S,2}^n\oplus Z_{2}^n),\nonumber
\end{align}
and
\begin{align}
H(U_{\bar{S}}\oplus Z^n)&=H(U_{\bar{S},2}^n\oplus Z_{2}^n)+H(U_{\bar{S},1}\oplus Z_1|U_{\bar{S},2}^n\oplus Z_{2}^n)\nonumber\\
&=H(U_{S,2}^n\oplus Z_{2}^n)+H(U_{\bar{S},1}\oplus Z_1|U_{\bar{S},2}^n\oplus Z_{2}^n)\nonumber
\end{align}
where the last equality follows from the fact that $P_{U_{S,2}^n\oplus Z_{2}^n}=P_{U_{\bar{S},2}^n\oplus Z_{2}^n}$ due to Lemma~\ref{lem:morebiass}. Thus, it suffices to show that
\begin{align}
H(U_{\bar{S},1}\oplus Z_1|U_{\bar{S},2}^n\oplus Z_{2}^n)\leq H(U_{S,1}\oplus Z_1|U_{S,2}^n\oplus Z_{2}^n).\nonumber
\end{align}
For any $\omega\in\{0,1\}^{n-1}$ let $\alpha_{\omega}\triangleq \Pr(U_{S,1}=1|U_{S,2}^n\oplus Z_{2}^n=\omega)$ and $\beta_{\omega}\triangleq \Pr(U_{\bar{S},1}=1|U_{\bar{S},2}^n\oplus Z_{2}^n=\omega)$. Then, we get
\begin{align}
&H(U_{\bar{S},1}\oplus Z_1|U_{\bar{S},2}^n\oplus Z_{2}^n)\nonumber\\
&=\sum_{\omega\in\{0,1\}^{n-1}}\Pr(U_{\bar{S},2}^n\oplus Z_{2}^n=\omega)h\left(\alpha*\beta_{\omega}\right)\nonumber\\
&=\sum_{\omega\in\{0,1\}^{n-1}}\Pr(U_{S,2}^n\oplus Z_{2}^n=\omega)h\left(\alpha*\beta_{\omega}\right)\nonumber\\
&\leq\sum_{\omega\in\{0,1\}^{n-1}}\Pr(U_{S,2}^n\oplus Z_{2}^n=\omega)h\left(\alpha*\alpha_{\omega}\right)\nonumber\\
&=H(U_{S,1}\oplus Z_1|U_{{S},2}^n\oplus Z_{2}^n),
\end{align}
where $a\ast b\triangleq a\cdot(1-b)+(1-a)\cdot b$ for any $a,b\in\left[0,1\right]$, the second equality follows since $P_{U_{S,2}^n\oplus Z_{2}^n}=P_{U_{\bar{S},2}^n\oplus Z_{2}^n}$, and the inequality is because $\beta_\omega$ is more biased than $\alpha_\omega$, by Lemma~\ref{lem:morebiass}.
\end{proof}

Applying Corollary~\ref{ref:shiftiteration} recursively, we see that for any $S\subset\ppp{0,1}^n$ we have
\begin{align}
H\left(U_{S_{\text{shifted}}}\oplus Z^n\right)\leq H(U_S\oplus Z^n).\label{eq:monotoneoptimal}
\end{align}
In fact, it is easy to extend the above argument to show that for any BMS channel with inputs $U_S$ and $U_{S_{\text{shifted}}}$ and corresponding outputs $Y^n$ and $\tilde{Y}^n$, respectively, we get $H(\tilde{Y}^n)\leq H(Y^n)$. Inequality~\eqref{eq:monotoneoptimal} immediately establishes Lemma~\ref{lem:monsuffice}.

We now turn to finding $H_m^n(\alpha)$ for $m=1,2,3,4$.

\begin{proposition}
$H_1^n(\alpha)=n\cdot h(\alpha)$.
\label{prop:H1}
\end{proposition}

\begin{proof}
For any vector $u\in\{0,1\}^n$ we have that $H(u\oplus Z^n)=H(Z^n)=n\cdot h(\alpha)$.
\end{proof}

\begin{proposition}
$H_2^n(\alpha)=1+(n-1)\cdot h(\alpha)$.
\label{prop:H2}
\end{proposition}

\begin{proof}
By Lemma~\ref{lem:monsuffice}, it is suffice to minimize $H(U_S\oplus Z^n)$ over $S\in\m{M}_2^n$. It is easy to see that $\m{M}_2^n$ consists of a single set $S^*=\{[1 \ 0 \cdots \ 0],[0 \ 0 \cdots \ 0]\}$, up to permuting the order of coordinates. Thus, direct calculation gives
\begin{align}
H_2^n(\alpha)&=H(U_{S^*}\oplus Z^n)=1+(n-1)\cdot h(\alpha).
\end{align}
\end{proof}
\vspace{-0.3cm}
\begin{proposition}
\begin{align}
H_3^n(\alpha)&=h\left(\frac{1}{3}*\alpha\right)+\left(\frac{2}{3}*\alpha\right)h\left(\frac{1-\alpha^2}{2-\alpha}\right)\nonumber\\
&+\left(\frac{1}{3}*\alpha\right)h\left(\frac{1-\alpha+\alpha^2}{1+\alpha}\right)+(n-2)h(\alpha)\label{eq:H3exact}\\
&\geq h\left(\frac{1}{3}*\alpha\right)+\frac{1}{3}h(\alpha)+\frac{2}{3}+(n-2)h(\alpha)\label{eq:H3bound}
\end{align}
\label{prop:H3}
\end{proposition}

\begin{proof}
By Lemma~\ref{lem:monsuffice}, it is suffice to minimize $H(U_S\oplus Z^n)$ over $S\in\m{M}_3^n$. It is easy to see that $\m{M}_3^n$ consists of a single set $S^*=\{[1 \ 0 \ 0  \ \cdots \ 0],[0 \ 1 \ 0 \ \cdots \ 0],[0 \ 0 \ 0 \ \cdots \ 0]\}$, up to permuting the order of coordinates. Thus,~\eqref{eq:H3exact} is obtained by direct calculation of $H(U_{S^*}\oplus Z^n)$. To obtain the lower bound~\eqref{eq:H3bound} we write
\begin{align}
&H_3^n(\alpha)=H(U_{S^*}\oplus Z^n)\nonumber\\
&=H(U_{S^*_1}\oplus Z_1)+H(U_{S^*_2}\oplus Z_2|U_{S^*_1}\oplus Z_1)+H\left( Z_3^n\right)\nonumber\\
&\geq  H(U_{S^*_1}\oplus Z_1)+H(U_{S^*_2}\oplus Z_2|U_{S^*_1})+H\left( Z_3^n\right)\nonumber\\
&=h\left(\frac{1}{3}*\alpha\right)+\frac{1}{3}h(\alpha)+\frac{2}{3}+(n-2)h(\alpha).\nonumber
\end{align}
\end{proof}

\begin{proposition}
$H_4^n(\alpha)=2+(n-2)\cdot h(\alpha)$.
\label{prop:H4}
\end{proposition}

\begin{proof}
By Lemma~\ref{lem:monsuffice}, it is suffice to minimize $H(U_S\oplus Z^n)$ over $S\in\m{M}_4^n$. It is easy to see that $\m{M}_4^n$ consists of two sets
\begin{align}
\m{C}\triangleq\{&[1 \ 1 \ 0  \ \cdots \ 0],[1 \ 0 \ 0 \ \cdots \ 0],\nonumber\\
&[0 \ 1 \ 0 \ \cdots \ 0],[0 \ 0 \ 0 \ \cdots \ 0]\},\nonumber\\
\m{B}\triangleq\{&[1 \ 0 \ 0  \ 0 \ \cdots \ 0],[0 \ 1 \ 0 \ 0 \ \cdots \ 0],\nonumber\\
&[0 \ 0 \  1 \ 0 \ \cdots \ 0],[0 \ 0 \ 0 \ 0  \ \cdots \ 0]\},\nonumber
\end{align}
up to permuting the order of coordinates. In particular, $\m{C}$ is the $2$-dimensional cube padded by $(n-2)$ zeros, whereas $\m{B}$ is the $3$-dimensional Hamming ball of radius $1$, padded by $n-3$ zeros. Thus,
\begin{align}
H_4^n(\alpha)=\min\left\{H\left(U_{\m{C}}\oplus Z^n\right),H\left(U_{\m{B}}\oplus Z^n\right)\right\}.\nonumber
\end{align}
It is easy to verify that $H\left(U_{\m{C}}\oplus Z^n\right)=2+(n-2)\cdot h(\alpha)$. We show that $H\left(U_{\m{B}}\oplus Z^n\right)\geq 2+(n-2)\cdot h(\alpha)$. Indeed,
\begin{align}
&H\left(U_{\m{B}}\oplus Z^n\right)\nonumber\\
&=H\left(U^2_{\m{B},1}\oplus Z_1^2\right)+H\left(U_{\m{B},3}\oplus Z_3|U^2_{\m{B},1}\oplus Z_1^2\right)+H(Z_4^n)\nonumber\\
&\geq H\left(U^2_{\m{B},1}\oplus Z_1^2\right)+H\left(U_{\m{B},3}\oplus Z_3|U^2_{\m{B},1}\right)+(n-3)\cdot h(\alpha)\nonumber\\
&=H\left(U^2_{\m{B},1}\oplus Z_1^2\right)+\frac{1}{2}+\frac{h(\alpha)}{2} +(n-3)\cdot h(\alpha).\label{eq:HBbound}
\end{align}
Direct calculation gives
\begin{align}
H\left(U^2_{\m{B},1}\oplus Z_1^2\right)=\frac{3}{2}+\frac{h(\alpha)}{2},
\end{align}
which together with~\eqref{eq:HBbound} shows that $H\left(U_{\m{B}}\oplus Z^n\right)\geq 2+(n-2)\cdot h(\alpha)$.
\end{proof}

We are now in a position to prove Lemma~\ref{lem:H2bound}.

\begin{proof}[Proof of Lemma~\ref{lem:H2bound}]
For any $m\geq 4$ we have that $H^n_m(\alpha)\geq H^n_4(\alpha)>H_1^n(\alpha)$, which implies that
\begin{align}
&\frac{m-2}{2m-2}H_1^n(\alpha)+\frac{m}{2m-2}H_m^n(\alpha)\geq \frac{H_1^n(\alpha)+H_m^n(\alpha)}{2}\nonumber\\
&\ \ \ \geq \frac{H_1^n(\alpha)+H_4^n(\alpha)}{2}=1+(n-1)\cdot h(\alpha)=H_2^n(\alpha).\nonumber
\end{align}
It then remains to verify~\eqref{H2Boundprove} for $m=3$. Using the lower bound~\eqref{eq:H3bound} for $H_3^n(\alpha)$, it suffices to verify that
\begin{align}
\frac{1}{4}nh(\alpha)&+\frac{3}{4}\pp{h\p{\frac{1}{3}*\alpha}+\frac{1}{3}h(\alpha)+\frac{2}{3}+(n-2)h(\alpha)}\nonumber\\
&\geq1+(n-1)h(\alpha),
\end{align}
which is equivalent to
\begin{align}
3\cdot h\p{(1/3)*\alpha}-2-h(\alpha)\geq0.
\end{align}
Let $g(\alpha)\triangleq3\cdot h\p{\frac{1}{3}*\alpha}-2-h(\alpha)$. It is easy to check that $g(0)>0$ and that $g(1/2)=0$. Thus, it suffices to show that $g(\alpha)$ is monotonically decreasing as a function of $\alpha$, namely, that $\mathrm{d}g(\alpha)/\mathrm{d}\alpha<0$, for any $\alpha\in(0,1/2)$. We have
\begin{align}
\frac{\mathrm{d}}{\mathrm{d}\alpha}g(\alpha) &= -\log_2\p{\frac{\frac{1}{3}\ast\alpha}{\frac{2}{3}\ast\alpha}}+\log_2\p{\frac{\alpha}{1-\alpha}}\\
& = \log_2\p{\frac{2\alpha-\alpha^2}{1-\alpha^2}},
\end{align}
which is negative for all $\alpha\in(0,1/2)$. 
\end{proof}

%
%
\vspace{-0.16cm}
\section*{Acknowledgment}
The authors are grateful to Yury Polyanskiy, Shlomo Shamai (Shitz), Ofer Shayevitz, and Omri Weinstein, for many discussions that helped prompt this work.

\bibliographystyle{IEEEtran}
\bibliography{OrBib2}

\end{document}